\documentclass[11pt]{article}
\pdfoutput=1

\usepackage{fullpage}
\usepackage{amsthm,amsmath,amssymb,amsfonts}
\usepackage{times}
\usepackage{microtype}
\usepackage{color}
\usepackage{colortbl}
\usepackage[table]{xcolor}
\usepackage[pdftex,colorlinks,citecolor=black,linkcolor=black,urlcolor=black,bookmarks=false]{hyperref}

\newif\ifnotes
\notestrue
\ifnotes
\newcommand{\alon}[1]{$\ll$\textsf{\color{green} Alon: { #1}}$\gg$}
\newcommand{\omer}[1]{$\ll$\textsf{\color{red} Omer: { #1}}$\gg$}
\newcommand{\nir}[1]{$\ll$\textsf{\color{orange} Nir: { #1}}$\gg$}
\newcommand{\ran}[1]{$\ll$\textsf{\color{blue} Ran: { #1}}$\gg$}
\newcommand{\yael}[1]{$\ll$\textsf{\color{purple} Yael: { #1}}$\gg$}
\newcommand{\henry}[1]{$\ll$\textsf{\color{brown} Henry: { #1}}$\gg$}
\newcommand{\shafi}[1]{$\ll$\textsf{\color{blue} Shafi: { #1}}$\gg$}
\else
\newcommand{\alon}[1]{}
\newcommand{\omer}[1]{}
\newcommand{\nir}[1]{}
\newcommand{\ran}[1]{}
\newcommand{\yael}[1]{}
\newcommand{\henry}[1]{}
\newcommand{\shafi}[1]{}
\fi

\renewcommand{\cal}{\mathcal}

\renewcommand{\paragraph}[1]{\vspace{1.5mm}\noindent \textbf{#1}}

\newcommand{\set}[1]{\left\{#1\right\}}

\newenvironment{boxfig}[2]{\begin{figure}[#1]\fbox{\begin{minipage}{\linewidth}
                        \vspace{0.2em}
                        \makebox[0.025\linewidth]{}
                        \begin{minipage}{0.95\linewidth}
            {{
                        #2 }}
                        \end{minipage}
                        \vspace{0.2em}
                        \end{minipage}}}{\end{figure}}

\newcommand{\pprotocol}[4]{
\begin{boxfig}{h!}{
\begin{center}
{\bf #1}
\end{center}
    #4
\vspace{0.2em} } \caption{\label{#3} #2}
\end{boxfig}
}

\newcommand{\protocol}[4]{
\pprotocol{#1}{#2}{#3}{#4} }

\newcommand{\poly}{\mathsf{poly}}
\newcommand{\negl}{\mathsf{negl}}
\newcommand{\aux}{\mathsf{aux}}
\newcommand{\Sim}{\mathsf{S}}
\newcommand{\Enc}{\mathsf{Enc}}
\newcommand{\Dec}{\mathsf{Dec}}
\newcommand{\Adv}{\mathsf{A}}
\newcommand{\Ver}{\mathsf{V}}
\newcommand{\ZK}{\mathsf{zk}}
\newcommand{\G}{\mathsf{G}}
\newcommand{\D}{\mathsf{D}}
\newcommand{\punc}{\mathsf{Punc}}
\newcommand{\rtime}{\mathsf{time}}

\newcommand{\pST}{\; \middle\vert \;}
\newcommand{\tC}{\tilde{C}}

\newcommand{\Z}{\mathcal{Z}}

\newcommand{\N}{\mathbb{N}}

\newcommand{\PPT}{\textrm{PPT}}
\renewcommand{\O}{\mathcal{O}}
\newcommand{\iO}{i\mathcal{O}}
\newcommand{\NP}{\mathsf{NP}}

\newtheorem{definition}{Definition}[section]
\newtheorem{lemma}[definition]{Lemma}
\newtheorem{corollary}[definition]{Corollary}
\newtheorem{theorem}[definition]{Theorem}
\newtheorem{claim}[definition]{Claim}

\theoremstyle{remark}

\newtheorem{remark}[definition]{Remark}

\begin{document}

\title{The Impossibility of Obfuscation with\\ Auxiliary Input or a Universal Simulator}

\author{ \and Nir Bitansky\thanks{Tel Aviv University, \texttt{nirbitan@tau.ac.il}.  Supported
by an IBM Ph.D. Fellowship, and the Check Point Institute for Information
Security.} \and %
Ran Canetti\thanks{Boston University and Tel Aviv University,
\texttt{canetti@bu.edu}.  Supported by the Check Point Institute for
Information Security, an NSF EAGER grant, and an NSF Algorithmic Foundations
grant 1218461.} \and %
Henry Cohn\thanks{Microsoft Research, One Memorial Drive,
Cambridge, MA 02142, \texttt{cohn@microsoft.com}.} \and %
Shafi Goldwasser\thanks{MIT and the Weizmann Institute of Science,
\texttt{shafi@theory.csail.mit.edu}.} \and \and %
Yael Tauman
Kalai\thanks{Microsoft Research, One Memorial Drive, Cambridge,
MA 02142, \texttt{yael@microsoft.com}.} \and %
Omer Paneth\thanks{Boston University, \texttt{omer@bu.edu}. Supported by the
Simons award for graduate students in theoretical computer
science and an NSF Algorithmic foundations grant 1218461.} \and %
Alon Rosen\thanks{Efi Arazi School of Computer Science, IDC Herzliya,
Israel, \texttt{alon.rosen@idc.ac.il}.  Supported by ISF grant no.\ 1255/12
and by the ERC under the EU's Seventh Framework Programme (FP/2007-2013) ERC
Grant Agreement n.\ 307952.} }

\date{February 9, 2014}

\maketitle

\thispagestyle{empty}

\begin{abstract}
In this paper we show that the existence of general indistinguishability
obfuscators conjectured in a few recent works implies, somewhat
counterintuitively, strong impossibility results for virtual black box
obfuscation. In particular, we show that indistinguishability obfuscation for all circuits implies:
\begin{itemize}
\item The impossibility of average-case virtual black box obfuscation with auxiliary
    input for any circuit family with super-polynomial pseudo-entropy.
    Such circuit families include all pseudo-random function families,
    and all families of encryption algorithms and randomized digital signatures that
    generate their required coin flips pseudo-randomly. Impossibility holds even when the auxiliary input depends only
    on the public circuit family, and not the specific circuit in the family being obfuscated.

\item The impossibility of average-case virtual black box obfuscation with a universal
    simulator (with or without any auxiliary input) for any circuit
    family with super-polynomial pseudo-entropy.
\end{itemize}
These bounds significantly strengthen the impossibility results of Goldwasser and Kalai (STOC 2005).
\end{abstract}

\newpage

\setcounter{page}{1}

\section{Introduction}\label{sec:intro}

The study of \emph{program obfuscation}---a method that transforms a program
(say, described as a Boolean circuit) into a form that is executable, but
otherwise completely unintelligible---has been a longstanding research
direction in cryptography. It was formalized by Barak
et~al.~\cite{BGIRSVY-conf}, who formulated a number of security notions for
this task. The strongest and most applicable of these notions is
\emph{virtual black box (VBB) obfuscation}, which requires that any adversary
trying to learn information from an obfuscated program cannot do better than
a simulator that is given only black-box access to the program. Barak
et~al.~constructed contrived function families that cannot be VBB obfuscated,
thus ruling out a universal obfuscator, but they left open the possibility
that large classes of programs might still be obfuscated. Subsequently, VBB
obfuscators were produced only for a number of restricted (and mostly simple)
classes of programs \cite{Canetti97,CD08,CanettiRV10,BR13b,BarakBCKPS13}. To date,
the classification of which programs can or cannot be VBB obfuscated is still
not well understood.

In contrast, recent progress for more relaxed notions of obfuscation suggests
a much more
positive picture: Garg et al.~\cite{GGHRSW13} proposed a candidate
construction for \emph{indistinguishability obfuscation} for \emph{all}
circuits.  This notion requires only that it is hard to distinguish an
obfuscation of $C_0$ from an obfuscation of $C_1$, where $C_0$ and $C_1$ are
circuits of the same size that compute the same function \cite{BGIRSVY-conf}.
Indeed, unlike the case of VBB obfuscation, there are no known impossibility
theorems for indistinguishability obfuscation. Furthermore, the Garg et al.\
construction and variants thereof were shown to satisfy the VBB guarantee in
ideal algebraic oracle models \cite{CanettiV13,BR13,BGKPS13}, although these
results have not proved useful so far in achieving VBB obfuscation in the
standard model of computation.

Although indistinguishability obfuscation might initially sound arcane, it is
surprisingly powerful.  For example, it amounts to \emph{best possible}
obfuscation \cite{GR07}, in the sense that anything that can be hidden by
some obfuscator will be hidden by every indistinguishability obfuscator.
Subsequent to \cite{GGHRSW13}, a flood of results have appeared showing that
indistinguishability obfuscation suffices for many applications, such as the
construction of public-key encryption from private-key encryption, the
existence of deniable encryption, the existence of multi-input functional
encryption, and more~\cite{SahaiW13,GGHRSW13,HSW13,GGJS13}.

Still, for many program classes the meaningfulness and applicability of
indistinguishability obfuscation is unclear. Thus, understanding which
classes of programs are VBB obfuscatable remains of central importance.
Aiming towards such a characterization, Goldwasser and Kalai \cite{GK05}
proved strong limitations on VBB obfuscation for a broad class of
\emph{pseudo-entropic programs}, including many cryptographic functions, such
as pseudo-random functions and certain natural instances of encryption and
signatures. They showed the impossibility of a form of VBB security with
respect to adversaries that have some a priori {\em auxiliary information}.
When the auxiliary information depends on the actual obfuscated program, they
showed that no class of pseudo-entropic functions can be obfuscated, assuming
VBB obfuscation for a simple class of \emph{point-filter functions}. For
auxiliary information that depends only on the class of programs to be
obfuscated, they gave an unconditional result, but only for a restricted
class of programs (those that evaluate \emph{NP-filter functions}).

\paragraph{This work in a nutshell.}
We strengthen the known impossibility results for VBB obfuscation with
auxiliary input, and we suggest a different, compelling interpretation of
auxiliary-input obfuscation. In a somewhat strange twist, our negative
results on VBB obfuscation are based on the existence of indistinguishability
obfuscation, which is typically viewed positively. Specifically:
\begin{itemize}
\item We weaken the conditions for the impossibility of \emph{dependent}
    auxiliary-input VBB obfuscation to \emph{witness encryption}, which in turn
    follows from \emph{indistinguishability obfuscation}.
\item We extend the impossibility of \emph{independent} auxiliary-input
    VBB obfuscation to \emph{all} pseudo-entropic functions, assuming
    \emph{indistinguishability obfuscation}.
\item We observe that  auxiliary-input  VBB obfuscation is equivalent to a very
    natural formulation of VBB obfuscation with universal simulation. This
    equivalence provides a clear conceptual argument for the significance
    of our extended impossibility results.
\end{itemize}

\medskip\noindent 
In the rest of the introduction, we introduce the notion of universal
simulation and further discuss the notion of auxiliary-input VBB obfuscation. Then, we
provide an overview of the results and sketch the proof techniques involved.

\paragraph{Universal simulators.}
The definition of VBB obfuscation as proposed by Barak et al. requires that
for each $\PPT$ adversary~$\Adv$, there exists a $\PPT$ simulator~$\Sim$ that
succeeds in simulating the output of $\Adv$ when $\Adv$ is given the
obfuscation $\O(f)$ but $\Sim$ is given only black-box access to~$f$. This
definition does not say how hard (or easy) it is to find the corresponding
simulator~$\Sim$ for a given adversary $\Adv$.
When security with black-box access to the function depends on computational
hardness assumptions, this definition leaves open the possibility that the
obfuscation could be broken in practice without providing an algorithm that
breaks these assumptions.

A stronger and arguably more meaningful definition requires that there exist an efficient transformation from an adversary to its corresponding simulator, or
equivalently a \emph{universal} $\PPT$ simulator capable of simulating any $\PPT$
adversary $\Adv$ given the code of $\Adv$. We will refer to such a definition
as VBB obfuscation with a \emph{universal simulator}.

As we said above, we will show that VBB obfuscation with a universal
simulator is impossible for function families with super-polynomial
pseudo-entropy if general indistinguishability obfuscation is possible.

\paragraph{Auxiliary input.}
The definition of VBB security with auxiliary inputs, originally considered
in~\cite{GK05}, is a strengthening of VBB security, which corresponds to a
setting in which the adversary may have some additional a priori information.

Allowing auxiliary input is crucial when obfuscation is used together with
other components in a larger scheme or protocol. Consider, for example, a
zero-knowledge protocol in which one of the prover's messages to the verifier
contains an obfuscated program $\O(f)$. To prove that the protocol is
zero-knowledge, we would like to show that every verifier $\Ver$ has a
zero-knowledge simulator $\Sim_{\ZK}$ that can simulate $\Ver$'s view of the
protocol. Intuitively, $\Sim_{\ZK}$ would rely on the security of $\O$ by
thinking of $\Ver$ as an ``obfuscation adversary'' that is trying to learn
information from $\O(f)$. Such an adversary has an ``obfuscation simulator''
$\Sim_{\O}$ that can learn the same information given only black-box access
to $f$, and $\Sim_{\ZK}$ can try to use $\Sim_{\O}$. The problem is that the
view of $\Ver$ does not depend only on the code of $\Ver$, but also on
auxiliary input to $\Ver$, such as other prover messages and the statement
being proven. An obfuscation definition that does not allow auxiliary input
is insufficient to handle this case.

The problem can be avoided by using a definition that guarantees the
existence of an obfuscation simulator that can simulate the view of $\Ver$
given any auxiliary input. If the obfuscated program $f$ depends on other
prover messages or on the statement, then we require security with respect to
\emph{dependent} auxiliary input.  Otherwise \emph{independent} auxiliary
input suffices. The paper~\cite{GK05} considered both of these notions. In
the case of dependent auxiliary input, the virtual black box property is
required to hold even when the auxiliary input given to the adversary and
simulator depends on the actual, secret circuit being obfuscated. In the case
of independent auxiliary input, this requirement is weakened: the auxiliary
input may depend only on the family of circuits, which is public. The actual
circuit to be obfuscated is chosen randomly from the family, independently of
the auxiliary input given to the adversary and simulator.

More precisely, an obfuscator $\O$ for a function family~${\cal F}$ is
(worst-case) VBB secure with \emph{dependent} auxiliary inputs if for every
probabilistic polynomial-time ($\PPT$) adversary $\Adv$, there exists a
$\PPT$ simulator $\Sim$ such that for every $f\in{\cal F}$ and every
auxiliary input~$\aux$ (which may depend on the function~$f$), the output of
$\Adv(\O(f),\aux(f))$ is computationally indistinguishable from
$\Sim^f(\aux(f))$.  The average-case analogue of this definition requires
that the output of $\Adv(\O(f),\aux(f))$ be computationally indistinguishable
from $\Sim^f(\aux(f))$ for a \emph{random} function $f\leftarrow{\cal F}$.

VBB security with \emph{independent} auxiliary inputs is defined only with
respect to an average-case definition.\footnote{It is not clear how to
enforce  that the auxiliary input is independent of the function in a
worst-case definition.} An obfuscator $\O$ for a function family~${\cal F}$
is average-case VBB secure with \emph{independent} auxiliary inputs if for
every $\PPT$ adversary $\Adv$, there exists a $\PPT$ simulator $\Sim$ such
that for every auxiliary input $\aux$ and for a random $f\leftarrow{\cal F}$,
the output of $\Adv(\O(f),\aux)$ is computationally indistinguishable from
$\Sim^f(\aux)$.

For the case of dependent auxiliary input, Goldwasser and Kalai \cite{GK05}
showed that functions with super-polynomial pseudo-entropy cannot be VBB
obfuscated, assuming that a different class of \emph{point filter functions}
can be VBB obfuscated. For the weaker notion of VBB obfuscation with
independent auxiliary input, they showed a more restricted impossibility
result for a subclass of functions called \emph{filter functions}.  Our
results extend these theorems, assuming indistinguishability obfuscators exist.

\subsection{Overview of results and techniques}

First we prove that VBB security with a
universal simulator is equivalent to VBB security with auxiliary inputs,
which is the obfuscation version of the known equivalence for zero-knowledge
proofs \cite{O87}. More specifically, we consider
both \emph{worst-case} VBB security and \emph{average-case} VBB security.  In
the former the simulator is required to successfully simulate the output
of~$\Adv$ \emph{for every} function in the family~${\cal F}$, whereas in the
latter the simulator is required to successfully simulate the output
of~$\Adv$ only for a \emph{random} function in the family.

We prove that worst-case VBB security with a universal simulator is
equivalent to worst-case VBB security with \emph{dependent} auxiliary
inputs, and that average-case VBB security with a universal simulator is
equivalent to average-case VBB security with \emph{independent} auxiliary
inputs.  To be consistent with the literature, when we refer to VBB security
we always consider the worst-case version.  When we would like to consider
the average-case version we refer to it as average-case VBB.

\paragraph{Informal Lemma~1.}
\emph{A candidate obfuscator is a (worst-case) VBB obfuscator with a
universal simulator for a class of functions~${\cal F}$  if and only if it
is a (worst-case) VBB obfuscator for ${\cal F}$ with dependent auxiliary
inputs.}

\paragraph{Informal Lemma~2.}
\emph{A candidate obfuscator is an average-case VBB obfuscator with a
universal simulator for a class of functions~${\cal F}$  if and only if it
is an average-case VBB obfuscator for ${\cal F}$ with independent auxiliary
inputs.}

\smallskip 

We state and prove these results as Lemmas~\ref{lemma:worst-case}
and~\ref{lemma:avg-case} in Section~\ref{section:equiv}.

The above two lemmas imply that in order to obtain negative results for VBB
obfuscation with a universal simulator, it suffices to obtain negative
results for VBB obfuscation with auxiliary inputs.

\paragraph{New impossibility results.}
We show that indistinguishability obfuscation implies that any function
family with super-polynomial pseudo-entropy \emph{cannot} be VBB obfuscated
with auxiliary input. Loosely speaking, a function family ${\cal F}$ has
super-polynomial pseudo-entropy if it is difficult to distinguish a genuine
function in ${\cal F}$ from one that has been randomly modified in some
locations: for every polynomial~$p$ there exists a polynomial-size set~$I$ of
inputs such that no efficient adversary can distinguish between a random
function $f\leftarrow{\cal F}$ and such a function with its values on $I$
replaced with another random variable with min-entropy~$p$.  We refer the
reader to Definition~\ref{def:pseudo-entropy} for the precise definition, but
note that such families include all pseudo-random function families.  They
also include all semantically secure secret-key or public-key encryption
schemes or secure digital signature schemes, provided that the randomness is
generated by using a (secret) pseudo-random function. (See Claim~4.0.1 in
\cite{GK05}.)

Recently, the notion of witness encryption was put forth by Garg
et~al.~\cite{GGSW13}. It was observed by Goldwasser et~al.~\cite{GKPVZ13}
that an extractable version of witness encryption can be used to obfuscate
the class of point-filter functions with respect to dependent auxiliary
inputs.  Thus, together with~\cite{GK05}, this shows that the existence of an
extractable witness encryption scheme implies that \emph{any} function with
super-polynomial pseudo-entropy cannot be obfuscated with respect to
dependent auxiliary inputs.

Here we show that the proof of~\cite{GK05} actually implies that witness
encryption, \emph{without} the extractability property, suffices to prove
that all functions with super-polynomial pseudo-entropy are not obfuscatable
with respect to dependent auxiliary inputs.

\paragraph{Informal Theorem 3.}
\emph{Assume the existence of a witness encryption scheme.  Then no function
family with super-polynomial pseudo-entropy has an average-case VBB
obfuscator with respect to dependent auxiliary input.}

\smallskip 

The idea behind the proof is that functions with high pseudo-entropy cannot
be efficiently compressed; i.e., given oracle access to such a function, one
cannot produce a small circuit for it.  The reason is that functions with
genuinely high entropy cannot be compressed at all (let alone efficiently),
and no efficient algorithm can distinguish them from those with high
pseudo-entropy.

Using this observation, the proof works as follows. Suppose we wish to
construct an obfuscation $\O(f)$ of a function $f$ that has high
pseudo-entropy on a polynomial-size set $I$ of inputs. We use witness
encryption to encrypt a random bit $b$ so that it can be read only by someone
who knows a circuit of size at most $|\O(f)|$ for the values of $f$ on $I$.
Given this encryption of $b$ as auxiliary input, knowledge of the circuit
$\O(f)$ suffices to decrypt $b$. However, black-box access to $f$ is not
enough to produce any small circuit, and so VBB security is violated.

We note that this theorem is true in the strong sense: for \emph{any} secret
predicate~$\pi(f)$ that is not learnable from black-box access to~$f$, there
exists an adversary and auxiliary input $\aux(f)$ such that given $\O(f)$ and
$\aux(f)$, the adversary efficiently recovers~$\pi(f)$, whereas given
$\aux(f)$ and oracle access to~$f$, it is computationally hard to
recover~$\pi(f)$. Moreover, the theorem holds even if we restrict $\aux(f)$
to be an efficiently computable function of $f$.

It was shown by Garg et~al.~\cite{GGSW13} (using different terminology) that
indistinguishability obfuscation for point-filter functions implies the
existence of witness encryption.  Thus, the informal theorem above can be
restated as follows:  assuming the existence of indistinguishability
obfuscation for point-filter functions, functions with super-polynomial
pseudo-entropy are not average-case VBB obfuscatable with respect to
dependent auxiliary inputs.

For independent auxiliary input, we use of a different hypothesis, namely
indistinguishability obfuscation for \emph{puncturable pseudo-random
functions} (see Definition~\ref{def:punc_prf}).  Roughly speaking, these are
pseudo-random functions for which we can produce alternate keys that
effectively randomize the output for a specified input while leaving the rest
of the function unchanged.

\paragraph{Informal Theorem 4.}
\emph{Assume the existence of indistinguishability obfuscation for a class of
puncturable pseudo-random functions. Then no function family with
super-polynomial pseudo-entropy has an average-case VBB obfuscator with
respect to independent auxiliary input.}

\smallskip 

The proof of this theorem is a little more subtle than the previous proof.
Suppose we are trying to obfuscate a circuit family with high pseudo-entropy
on a set $I$ of inputs.  The auxiliary input will be $\iO(K_s)$, where $\iO$
denotes indistinguishability obfuscation and $K_s$ is a circuit that takes
another circuit $\tC$ as input and applies a puncturable pseudo-random
function $\G_s$ to the values $\tC(I)$ of $\tC$ on $I$. Here, $s$ is a random
key.

Now, let $\O(C)$ be a candidate obfuscation of a circuit $C$.  By definition,
applying the auxiliary circuit $\iO(K_s)$ to $\O(C)$ yields $K_s(C)$ (i.e.,
$\G_s(C(I))$), but we will show that $K_s(C)$ cannot be computed using only
black-box access to $C$.  If it could, then we could replace the $C$ oracle
with suitable random values $Y$ on $I$ and still get the answer $\G_s(Y)$, by
the definition of pseudo-entropy. Then we could modify the auxiliary input to
be $\iO(K^*_s)$, where the pseudo-random function in $K^*_s$ has been punctured
to randomize its value at $Y$.  The reason this modification is allowable is
that with high probability, $K_s$ and $K^*_s$ define the same function ($Y$
has entropy too high to be compressible to any small circuit, so no input
$\tC$ to $K^*_s$ will ever satisfy $\tC(I) = Y$).  Thus, $\iO(K_s)$ and
$\iO(K^*_s)$ are indistinguishable. However, by construction $K^*_s$ does not
determine the value $\G_s(Y)$, which is a contradiction.

We state and prove these results more formally as Theorems~\ref{thm:main1}
and~\ref{thm:main2}.  Together with Lemmas~\ref{lemma:worst-case}
and~\ref{lemma:avg-case}, they immediately yield impossibility results for
VBB obfuscation with a universal simulator. In particular,
Theorem~\ref{thm:main1} and Lemma~\ref{lemma:worst-case} imply the following
corollary.

\paragraph{Corollary~1.}
\emph{Assume the existence of a witness encryption scheme.  Then no function
family with super-polynomial pseudo-entropy has a VBB obfuscator with a
universal simulator.}

\smallskip 

As was the case for Theorem~\ref{thm:main1}, this corollary is true in the
strong sense: for \emph{any} secret predicate~$\pi(f)$ that is not learnable
from black-box access to~$f$, there exists an adversary that efficiently
recovers~$\pi(f)$ given $\O(f)$, whereas given the code of the adversary and
given oracle access to~$f$, it is computationally hard to recover~$\pi(f)$.

Theorem~\ref{thm:main2} and Lemma~\ref{lemma:avg-case} imply the following
corollary.

\paragraph{Corollary~2.}
\emph{Assume the existence of indistinguishability obfuscation for a class
of puncturable pseudo-random functions. Then no function family with
super-polynomial pseudo-entropy has an average-case VBB obfuscator with a
universal simulator.}

\section{Preliminaries}

Let ${\cal F}=\{f_s\}$ be a family of polynomial-size circuits. In what
follows, we write ${\cal F}=\bigcup_{k\in\N} {\cal F}_k$ with ${\cal
F}_k=\{f_s\}_{s\in\{0,1\}^k}$. Each circuit $f_s$ will have size
$\poly(|s|)$, where $\poly$ denotes an unspecified, polynomially-bounded
function.

\begin{definition}[\bf VBB obfuscation with universal simulator]\label{def:VBB}
Let ${\cal F}=\{f_s\}$ be a family of polynomial-size circuits.  We say that
a probabilistic algorithm $\O$ (mapping circuits to circuits) is an
obfuscation of ${\cal F}$ with a universal simulator if the following
conditions hold:
\begin{itemize}
\item \textbf{Correctness:} For every function $f_s\in{\cal F}$ and every
    possible input $x$,
\[
\O(f_s)(x)=f_s(x).
\]
I.e., the random variable $\O(f_s)$ defines the same function as $f_s$
with probability $1$.

\item \textbf{Polynomial slowdown:} There exists a polynomial~$p$ such
    that for every $f_s\in{\cal F}$, \[|\O(f_s)|\leq p(|f_s|).\]

\item \textbf{Security with a universal simulator:} There exists a (possibly non-uniform)
    $\PPT$~$\Sim$ such that for every (possibly non-uniform)
    $\PPT$~$\Adv$, every predicate~$\pi$, every $k\in\N$, and every
    $s\in\{0,1\}^k$,
    \begin{equation}\label{eqn:security}
    \left|\Pr[\Adv(\O(f_s))=\pi(s)]-\Pr[\Sim^{f_s}(\Adv)=\pi(s)]\right|=\negl(k),
    \end{equation}
    where the probabilities are over the random coin tosses of~$\Adv$ and
    $\Sim$.  Here $\negl(k)$ denotes an unspecified, negligible function
    (i.e., $|\negl(k)| = O(1/k^c)$ for each constant $c>0$).
\end{itemize}
We say that $\O$ is an {\bf average-case} obfuscation of ${\cal F}$ with a
universal simulator if Equation~\eqref{eqn:security} holds for {\bf random}
$s\leftarrow\{0,1\}^k$; in other words, it means there exists a (possibly
non-uniform) $\PPT$~$\Sim$ such that for every (possibly non-uniform)
$\PPT$~$\Adv$, every predicate~$\pi$, and every $k\in\N$,
\[
\left|\Pr[\Adv(\O(f_s))=\pi(s)]-\Pr[\Sim^{f_s}(\Adv)=\pi(s)]\right|=\negl(k),
\]
where the probabilities are over $s\leftarrow \{0,1\}^k$ and over the random
coin tosses of~$\Adv$ and $\Sim$.
\end{definition}

Note that we do not assume $\O(f_s)$ can be efficiently computed given $f_s$.
Our negative results rule out the existence of obfuscations, and not merely
the possibility of finding them.

When $\Adv$ is non-uniform, the notation $\Sim^{f_s}(\Adv)$ of course means
that $\Sim$ is given a circuit for $\Adv$ for inputs of the appropriate size.
When $\Adv$ is uniform, it means the same thing as in the non-uniform case;
equivalently, $\Sim$ is given the code for $\Adv$ together with
$1^{\rtime(\Adv(\O(f_s)))}$ to ensure that it is allowed enough time.

In Definition~\ref{def:VBB}, we have conflated the circuit size parameter $k$
and the security parameter of the obfuscation method.  One could distinguish
between them at the cost of more notation, but this conflation is of course
harmless for proving impossibility theorems.

\begin{definition}[\bf VBB obfuscation with auxiliary inputs]
Let ${\cal F}=\{f_s\}$ be a family of polynomial-size circuits.  We say that
a probabilistic algorithm $\O$ is an obfuscation of ${\cal F}$ with
(dependent) auxiliary inputs if it satisfies the correctness and polynomial
slowdown conditions of Definition~\ref{def:VBB}, and in addition it satisfies
the following security requirement:
\begin{itemize}
\item \textbf{Security with auxiliary inputs:}  For every (possibly
    non-uniform) $\PPT$ $\Adv$, there exists a (possibly non-uniform)
    $\PPT$ $\Sim$ such that for every predicate~$\pi$, every $k\in\N$,
    every $s\in\{0,1\}^k$, and every auxiliary input $\aux(s)$ of size
    $\poly(k)$,
    \begin{equation}\label{eqn:aux-security}
    \left|\Pr[\Adv(\O(f_s),\aux(s))=\pi(s,\aux(s))]-\Pr[\Sim^{f_s}(\aux(s))=\pi(s,\aux(s))]\right|=\negl(k),
    \end{equation}
where the probabilities are over the random coin tosses of~$\Adv$ and
$\Sim$.  We write $\aux(s)$ as a function of $s$ for clarity, but this is
not strictly necessary since the quantification automatically allows
dependence on $s$.
\end{itemize}
We say that $\O$ is an {\bf average-case} obfuscation of ${\cal F}$ with
(dependent) auxiliary inputs if Equation~\eqref{eqn:aux-security} holds for
{\bf random} $s\leftarrow\{0,1\}^k$; namely, if for every (possibly
non-uniform) $\PPT$ $\Adv$ there exists a (possibly non-uniform) $\PPT$
$\Sim$ such that for every predicate~$\pi$, every $k\in\N$, and every
auxiliary input $\aux(s)$ of size $\poly(s)$ (and allowed to depend on $s$),
\[
\left|\Pr[\Adv(\O(f_s),\aux(s))=\pi(s,\aux(s))]-\Pr[\Sim^{f_s}(\aux(s))=\pi(s,\aux(s))]\right|=\negl(k),
\]
where the probabilities are over $s\leftarrow \{0,1\}^k$ and over the random
coin tosses of~$\Adv$ and $\Sim$.
\end{definition}

In the definition above we allowed the auxiliary input to depend on the
function being obfuscated.  In what follows we define VBB obfuscation with
\emph{independent} auxiliary inputs, where we restrict the auxiliary input
to be \emph{independent} of the function being obfuscated.  For this
definition, only the average-case version makes sense, since in the
worst-case version it is not clear how to ensure that the auxiliary input is
independent of the function being obfuscated.

\begin{definition}[\bf Average-case VBB obfuscation with independent auxiliary inputs]
Let ${\cal F}=\{f_s\}$ be a family of polynomial-size circuits.  We say that
$\O$ is an obfuscation of ${\cal F}$ with independent auxiliary inputs if it
satisfies the correctness and polynomial slowdown conditions of
Definition~\ref{def:VBB}, and in addition it satisfies the following security
requirement:
\begin{itemize}
\item \textbf{Average-case security with independent auxiliary input:}
    For every (possibly non-uniform) $\PPT$ $\Adv$, there exists a
    (possibly non-uniform) $\PPT$ $\Sim$ such that for every
    predicate~$\pi$, every $k\in\N$, and every auxiliary input
    $\aux\in\{0,1\}^{\poly(k)}$,
    \begin{equation*}
    \left|\Pr[\Adv(\O(f_s),\aux)=\pi(s,\aux)]-\Pr[\Sim^{f_s}(\aux)=\pi(s,\aux)]\right|=\negl(k),
    \end{equation*}
where the probabilities are over $s\leftarrow\{0,1\}^k$ and over the
random coin tosses of~$\Adv$ and $\Sim$.
\end{itemize}

\end{definition}

\begin{definition}[\bf Witness encryption]
A witness encryption scheme for an $\NP$ language ${\cal L}$ with
corresponding witness relation ${\cal R}_{\cal L}$ is a pair of $\PPT$
algorithms $(\Enc,\Dec)$ such that the following conditions hold:
\begin{itemize}
\item \textbf{Correctness:} For all $(x, w) \in {\cal R}_{\cal L}$ and
    every $b\in\{0,1\}$,
\[
\Pr[\Dec(\Enc_x(1^k,b),w) = b] = 1 - \negl(k).
\]

\item\textbf{Semantic Security:} For every $x\not\in{\cal L}$ and every
    (possibly non-uniform) $\PPT$ adversary $\Adv$,
\[
\left|\Pr[\Adv(\Enc_x(1^k,0))=1]-\Pr[\Adv(\Enc_x(1^k,1))=1]\right|=\negl(k),
\]
where the probability is over the random coin tosses of $\Enc$ and
$\Adv$.
\end{itemize}
\end{definition}

\begin{definition}[\bf Indistinguishability obfuscation]\label{def:iO}
Let ${\cal C}$ be a family of polynomial-size circuits. A $\PPT$ algorithm
$\iO$ is said to be an indistinguishability obfuscator for $\cal{C}$ if it
satisfies the correctness and polynomial slowdown conditions of
Definition~\ref{def:VBB}, and in addition it satisfies the following security
requirement:
\begin{itemize}
\item \textbf{Indistinguishability:} For all $C,C' \in \cal{C}$ that are
    of the same size and define the same function, $\iO(C)$ and $\iO(C')$
    are computationally indistinguishable. More formally, for every
    (possibly non-uniform) $\PPT$ distinguisher $\D$,
    \[
     \left|\Pr[\D(\iO(C))=1]-\Pr[\D(\iO(C'))=1]\right|= \negl(k),
    \]
    where the probability is over the random coin tosses of $\iO$ and
    $\D$.
\end{itemize}
\end{definition}

Although Definition~\ref{def:VBB} did not require VBB obfuscation to be
efficiently computable (to obtain stronger impossibility results),
we require $\iO$ to be efficiently computable in Definition~\ref{def:iO}, because inefficient
indistinguishability obfuscation is trivial.

We next define puncturable pseudo-random functions. We consider a simple case
in which any PRF might be punctured at a single point. The definition is
formulated as in \cite{SahaiW13}.

\begin{definition}[\bf Puncturable PRFs]\label{def:punc_prf}
Let $\ell,m$ be polynomially bounded length functions. An efficiently
computable family of functions
\begin{align*}
\mathcal{G} = \set{\G_s\colon \{0,1\}^{m(k)}\rightarrow\{0,1\}^{\ell(k)}\pST s \in \{0,1\}^{k}, k\in\N},
\end{align*}
associated with an efficient (probabilistic) key sampler
$\mathsf{Gen}_\mathcal{G}$, is a puncturable PRF if there exists a puncturing
algorithm $\punc$ that takes as input a key $s\in \{0,1\}^k$ and a point $x^*
\in \{0,1\}^{m(k)}$ and outputs a punctured key $s_{x^*}$ so that the
following conditions are satisfied:
\begin{itemize}
\item \textbf{Functionality is preserved under puncturing:} For every
    $x^*\in \{0,1\}^{m(k)}$, if we sample $s$ from
    $\mathsf{Gen}_\mathcal{G}(1^k)$ and let $s_{x^*} =  \punc(s,x^*)$,
    then $\G_s$ and $\G_{s_{x^*}}$ have the same values at every point
    other than $x^*$ with probability $1$.

\item \textbf{Indistinguishability at punctured points:} The two
    ensembles
\begin{align*}
\set{\big(x^*,s_{x^*},\G_s(x^*)\big)\pST s \gets \mathsf{Gen}_\mathcal{G}(1^k),
    s_{x^*}= \punc(s,x^*)}_{x^*\in \{0,1\}^{m(k)},k\in \N},\\
\set{\big(x^*,s_{x^*}, u \big) \pST s \gets \mathsf{Gen}_\mathcal{G}(1^k),
    s_{x^*}=\punc(s,x^*),u \gets \{0,1\}^{\ell(k)}}_{x^*\in
    \{0,1\}^{m(k)},k\in \N}
\end{align*}
are computationally indistinguishable by (possibly non-uniform) $\PPT$
distinguishers.
\end{itemize}
\end{definition}
To be explicit, we include $x^*$ in the distribution; throughout, we shall
assume for simplicity that a punctured key $s_{x^*}$ includes $x^*$ in the
clear. As shown in \cite{BoyleGI13,BonehW13,KiayiasPTZ13}, the pseudo-random
functions from \cite{GoldreichGM86} yield puncturable PRFs as defined above.

\begin{definition}[\bf Pseudo-entropy of a circuit class]\label{def:pseudo-entropy}
Let $p=p(k)$  be a polynomial.  We say that a class of circuits ${\cal
C}=\bigcup_{k\in\N} {\cal C}_k$ has pseudo-entropy at least $p=p(k)$, if
there exists a polynomial~$t=t(k)$ and a subset $I_k\subseteq\{0,1\}^k$ of
size $t(k)$, and for every $C\in {\cal C}_k$ there exists a random variable
$Y^C=(Y_i)_{i \in I_k} \in \{0,1\}^{I_k}$, such that the following conditions
hold:
 \begin{enumerate}
 \item The random variable $Y^C$ has statistical min-entropy at least
     $p(k)$.  In other words, each of its values occurs with probability
     at most $2^{-p(k)}$.
 \item For every (possibly non-uniform) $\PPT$ distinguisher $\D$,
     \[
     \left|\Pr[\D^{C}(1^k)=1]-\Pr[\D^{C \circ Y^C}(1^k)=1]\right|= \negl(k),
     \]
     where $C \circ Y^C$ denotes an oracle that agrees with $C$ except
that $Y^C$ replaces the values of $C$ for inputs in $I_k$.  Here the
probabilities are over $C\leftarrow{\cal C}_k$, the random variable
$Y^C$, and the random coin tosses of $\D$.
 \end{enumerate}

We say that ${\cal C}$ has super-polynomial pseudo-entropy if it has
pseudo-entropy at least~$p$ for every polynomial~$p$, and we then call the
circuits in ${\cal C}$ pseudo-entropic.
\end{definition}

\section{Equivalence between a universal simulator and auxiliary inputs}
\label{section:equiv}
In this section we show that VBB obfuscation with a universal simulator is
equivalent to VBB obfuscation with auxiliary inputs. Specifically, we prove
the following two lemmas.

\begin{lemma}\label{lemma:worst-case}
Let ${\cal F}=\{f_s\}$ be a family of polynomial-size circuits. Then $\O$ is
a VBB obfuscator for ${\cal F}$ with a universal simulator if and only if it
is a VBB obfuscator for ${\cal F}$ with dependent auxiliary inputs.
\end{lemma}

\begin{lemma}\label{lemma:avg-case}
Let ${\cal F}=\{f_s\}$ be a family of polynomial-size circuits. Then $\O$ is
an average-case VBB obfuscator for ${\cal F}$ with a universal simulator if
and only if it is an average-case VBB obfuscator for ${\cal F}$ with
independent auxiliary inputs.
\end{lemma}

\paragraph{Proof of Lemma~\ref{lemma:worst-case}.}

\noindent $(\Rightarrow)$:  Suppose that $\O$ is a VBB obfuscator for ${\cal
F}$ with a universal simulator.  Namely, there exists a (possibly
non-uniform) $\PPT$~$\Sim$ such that for every (possibly non-uniform)
$\PPT$~$\Adv$, every predicate~$\pi$, every $k\in\N$ and every
$s\in\{0,1\}^k$,
\[
\left|\Pr[\Adv(\O(f_s))=\pi(s)]-\Pr[\Sim^{f_s}(\Adv)=\pi(s)]\right|=\negl(k),
\]
where the probabilities are over the random coin tosses of $\Adv$
and~$\Sim$.

We will prove that $\O$ is a VBB obfuscator for ${\cal F}$ with dependent
auxiliary inputs.  To this end, fix any (possibly non-uniform) $\PPT$
adversary~$\Adv$.  Let $\Sim_{\Adv}$ be the $\PPT$ simulator defined as
follows:  for every auxiliary input~$\aux(s)$, $\Sim_{\Adv}^{f_s}(\aux(s))$
runs the universal simulator~$\Sim^{f_s}$ on input $\Adv_{\aux(s)}$, where
$\Adv_{\aux(s)}$ is the (non-uniform) adversary that simulates~$\Adv$ with
auxiliary input~$\aux(s)$. We need to prove that for every predicate~$\pi$,
every $k\in\N$, and every $s\in\{0,1\}^k$,
\[
\left|\Pr[\Adv(\O(f_s),\aux(s))=\pi(s,\aux(s))]-\Pr[\Sim_\Adv^{f_s}(\aux(s))=\pi(s,\aux(s))]\right|=\negl(k),
\]
where the probabilities are over the random coin tosses of~$\Adv$
and~$\Sim$.

To do so, we check that
\begin{align*}
&\left|\Pr[\Adv(\O(f_s),\aux(s))=\pi(s,\aux(s))]-\Pr[\Sim_\Adv^{f_s}(\aux(s))=\pi(s,\aux(s))]\right|\\
& \qquad = \left|\Pr[\Adv(\O(f_s),\aux(s))=\pi(s,\aux(s))]-\Pr[\Sim^{f_s}(\Adv_{\aux(s)})=\pi(s,\aux(s))]\right|\\
& \qquad\leq\left|\Pr[\Adv(\O(f_s),\aux(s))=\pi(s,\aux(s))]-\Pr[\Adv_{\aux(s)}(\O(f_s))=\pi(s,\aux(s))]\right|\\
& \qquad \quad\  \phantom{} + \left|\Pr[\Adv_{\aux(s)}(\O(f_s))=\pi(s,\aux(s))]-\Pr[\Sim^{f_s}(\Adv_{\aux(s)})=\pi(s,\aux(s))]\right|\\
& \qquad =\negl(k),
\end{align*}
where the first equation follows by the definition of~$\Sim_\Adv$, the
inequality follows from the triangle inequality, and the last equation
follows from the definition of~$\Adv_{\aux(s)}$ and from the fact that $\O$
is
VBB secure with the universal simulator~$\Sim$.\\ 

\noindent $(\Leftarrow)$:  Suppose that $\O$ is a VBB obfuscator for ${\cal
F}$ with dependent auxiliary inputs.  Namely, for every (possibly
non-uniform) $\PPT$ $\Adv$ there exists a (possibly non-uniform)
$\PPT$~$\Sim$ such that for every predicate~$\pi$, every $k\in\N$, every
$s\in\{0,1\}^k$, and every auxiliary input $\aux(s)$ of size $\poly(k)$,
\[
\left|\Pr[\Adv(\O(f_s),\aux(s))=\pi(s,\aux(s))]-\Pr[\Sim^{f_s}(\aux(s))=\pi(s,\aux(s))]\right|=\negl(k),
\]
where the probabilities are over the random coin tosses of~$\Adv$ and~$\Sim$.
We prove that~$\O$ is a VBB obfuscator for ${\cal F}$ with a universal
simulator.  To this end, let $\Adv^*$ be a universal $\PPT$ adversary that
interprets its auxiliary input~$\aux=\aux(s)$ as a (possibly non-uniform)
$\PPT$ adversary and runs this adversary.
(As pointed out after Definition~\ref{def:VBB}, we must interpret this
carefully regarding running times in the uniform case.) The fact that~$\O$ is
a VBB obfuscator with dependent auxiliary inputs implies that there is a
$\PPT$ simulator~$\Sim$ such that for every predicate~$\pi$, every $k\in\N$,
every $s\in\{0,1\}^k$, and every auxiliary input $\aux(s)$ of size
$\poly(k)$,
\begin{equation}\label{eqn:A}
\left|\Pr[\Adv^*(\O(f_s),\aux(s))=\pi(s,\aux(s))]-\Pr[\Sim^{f_s}(\aux(s))=\pi(s,\aux(s))]\right|=\negl(k),
\end{equation}
where the probabilities are over the random coin tosses of~$\Adv^*$
and~$\Sim$. We claim that $\Sim$ is a universal simulator for $\O$. Namely,
we claim that for every (possibly non-uniform) $\PPT$ adversary~$\Adv$, every
predicate~$\pi$, every $k\in\N$, and every $s\in\{0,1\}^k$,
\[
\left|\Pr[\Adv(\O(f_s))=\pi(s)]-\Pr[\Sim^{f_s}(\Adv)=\pi(s)]\right|=\negl(k).
\]
To see why, note that
\begin{align*}
&\left|\Pr[\Adv(\O(f_s))=\pi(s)]-\Pr[\Sim^{f_s}(\Adv)=\pi(s)]\right|\\
&\qquad \leq \left|\Pr[\Adv(\O(f_s))=\pi(s)]-\Pr[\Adv^*(\O(f_s),\Adv)=\pi(s)]\right|\\
&\qquad \quad\  \phantom{} + \left|\Pr[\Adv^*(\O(f_s),\Adv)=\pi(s)]-\Pr[\Sim^{f_s}(\Adv)=\pi(s)]\right|\\
&\qquad = \left|\Pr[\Adv^*(\O(f_s),\Adv)=\pi(s)]-\Pr[\Sim^{f_s}(\Adv)=\pi(s)]\right|\\
&\qquad = \negl(k),
\end{align*}
where the inequality follows from the triangle inequality, the next equation
follows from the definition of~$\Adv^*$, and the last equation follows from
Equation~\eqref{eqn:A}. \qed

\medskip
\paragraph{Proof of Lemma~\ref{lemma:avg-case}.}

\noindent $(\Rightarrow)$:  Suppose that $\O$ is an average-case VBB
obfuscator for ${\cal F}$ with a universal simulator.  Namely, there exists a
(possibly non-uniform) $\PPT$~$\Sim$ such that for every (possibly
non-uniform) $\PPT$~$\Adv$, every predicate~$\pi$, and every $k\in\N$,
\[
\left|\Pr[\Adv(\O(f_s))=\pi(s)]-\Pr[\Sim^{f_s}(\Adv)=\pi(s)]\right|=\negl(k),
\]
where the probabilities are over $s\leftarrow\{0,1\}^k$ and over the random
coin tosses of $\Adv$ and~$\Sim$.

We will prove that $\O$ is an average-case VBB obfuscator for ${\cal F}$ with
independent auxiliary inputs.  To this end, fix any (possibly non-uniform)
$\PPT$ adversary~$\Adv$.  Let $\Sim_{\Adv}$ be the $\PPT$ simulator defined
as follows:  for every auxiliary input~$\aux$, $\Sim_{\Adv}^{f_s}(\aux)$ runs
the universal simulator~$\Sim^{f_s}$ on input $\Adv_{\aux}$, where
$\Adv_{\aux}$ is the (non-uniform) adversary that simulates~$\Adv$ with
auxiliary input~$\aux$. We need to prove that for every predicate~$\pi$,
every $k\in\N$, and every $\aux\in\{0,1\}^{\poly(k)}$,
\[
\left|\Pr[\Adv(\O(f_s),\aux)=\pi(s,\aux)]-\Pr[\Sim_\Adv^{f_s}(\aux)=\pi(s,\aux)]\right|=\negl(k),
\]
where the probabilities are over $s\leftarrow\{0,1\}^k$ and over the random
coin tosses of~$\Adv$ and~$\Sim$. To see why this is true, note that
\begin{align*}
&\left|\Pr[\Adv(\O(f_s),\aux)=\pi(s,\aux)]-\Pr[\Sim_\Adv^{f_s}(\aux)=\pi(s,\aux)]\right|\\
&\qquad =\left|\Pr[\Adv(\O(f_s),\aux)=\pi(s,\aux)]-\Pr[\Sim^{f_s}(\Adv_{\aux})=\pi(s,\aux)]\right|\\
&\qquad \leq\left|\Pr[\Adv(\O(f_s),\aux)=\pi(s,\aux)]-\Pr[\Adv_{\aux}(\O(f_s))=\pi(s,\aux)]\right|\\
&\qquad \quad\ \phantom{}+\left|\Pr[\Adv_{\aux}(\O(f_s))=\pi(s,\aux)]-\Pr[\Sim^{f_s}(\Adv_{\aux})=\pi(s,\aux)]\right|\\
&\qquad =\negl(k),
\end{align*}
where the first equation follows from the definition of~$\Sim_\Adv$, the
inequality follows from the triangle inequality, and the last equation
follows from the definition of~$\Adv_{\aux}$ and from the fact that
$\O$ is average-case VBB secure with the universal simulator~$\Sim$.\\

\noindent $(\Leftarrow)$:  Suppose that $\O$ is an average-case VBB
obfuscator for ${\cal F}$ with independent auxiliary inputs.  Namely, for
every (possibly non-uniform) $\PPT$ $\Adv$, there exists a (possibly
non-uniform) $\PPT$~$\Sim$ such that for every predicate~$\pi$, every
$k\in\N$, and every auxiliary input $\aux\in\{0,1\}^{\poly(k)}$,
\[
\left|\Pr[\Adv(\O(f_s),\aux)=\pi(s,\aux)]-\Pr[\Sim^{f_s}(\aux)=\pi(s,\aux)]\right|=\negl(k),
\]
where the probabilities are over $s\leftarrow\{0,1\}^k$ and over the random
coin tosses of~$\Adv$ and~$\Sim$.

We will prove that~$\O$ is an average-case VBB obfuscator for ${\cal F}$ with
a universal simulator.  To this end, let~$\Adv^*$ be a universal $\PPT$
adversary that interprets its auxiliary input~$\aux$ as a (possibly
non-uniform) $\PPT$ adversary and runs this adversary,
as in the previous proof. The fact that~$\O$ is an average-case VBB
obfuscator with independent auxiliary inputs implies that there is a $\PPT$
simulator~$\Sim$ such that for every predicate~$\pi$, every $k\in\N$, and
every auxiliary input~$\aux\in\{0,1\}^{\poly(k)}$,
\begin{equation}\label{eqn:B}
\left|\Pr[\Adv^*(\O(f_s),\aux)=\pi(s,\aux)]-\Pr[\Sim^{f_s}(\aux)=\pi(s,\aux)]\right|=\negl(k),
\end{equation}
where the probabilities are over $s\leftarrow\{0,1\}^k$ and over the random
coin tosses of~$\Adv^*$ and~$\Sim$. We claim that $\Sim$ is an average-case
universal simulator for $\O$.  Namely, we claim that for every (possibly
non-uniform) $\PPT$ adversary~$\Adv$, every predicate~$\pi$, and every
$k\in\N$,
\[
\left|\Pr[\Adv(\O(f_s))=\pi(s)]-\Pr[\Sim^{f_s}(\Adv)=\pi(s)]\right|=\negl(k),
\]
where the probabilities are over $s\leftarrow\{0,1\}^k$, and over the random
coin tosses of~$\Adv$ and~$\Sim$.

To see why, note that
\begin{align*}
&\left|\Pr[\Adv(\O(f_s))=\pi(s)]-\Pr[\Sim^{f_s}(\Adv)=\pi(s)]\right|\\
&\qquad \leq\left|\Pr[\Adv(\O(f_s))=\pi(s)]-\Pr[\Adv^*(\O(f_s),\Adv)=\pi(s)]\right|\\
&\qquad \quad\ \phantom{}+\left|\Pr[\Adv^*(\O(f_s),\Adv)=\pi(s)]-\Pr[\Sim^{f_s}(\Adv)=\pi(s)]\right|\\
&\qquad =\left|\Pr[\Adv^*(\O(f_s),\Adv)=\pi(s)]-\Pr[\Sim^{f_s}(\Adv)=\pi(s)]\right|\\
&\qquad =\negl(k),
\end{align*}
where the inequality follows from the triangle inequality, the next equation
follows from the definition of~$\Adv^*$, and the last equation follows from
Equation~\eqref{eqn:B}.

\qed

\section{Impossibility for obfuscation with auxiliary inputs}

As mentioned in the introduction, Goldwasser and Kalai~\cite{GK05} proved
that either point-filter functions are not obfuscatable with dependent
auxiliary inputs or \emph{all} function families with sufficient
pseudo-entropy are not obfuscatable with dependent auxiliary inputs. It was
recently observed by Goldwasser et~al.~\cite{GKPVZ13} that extractable
witness encryption implies that point-filter functions are obfuscatable with
dependent auxiliary inputs, and thus that any function family with
sufficient pseudo-entropy is not obfuscatable with dependent auxiliary
inputs. We now show that the same impossibility result  (with essentially
the same proof as in~\cite{GK05}) can be obtained assuming the existence of
witness encryption, without any extractability property.

\begin{theorem}\label{thm:main1}
Assume the existence of a witness encryption scheme for an $\NP$-complete
language. Then no function family with super-polynomial pseudo-entropy has an
average-case VBB obfuscator with respect to dependent auxiliary input.
\end{theorem}

In fact, the proof rules out average-case obfuscation if we restrict the
auxiliary input to be efficiently computable given the function (or even
oracle access to the function).

\begin{theorem}\label{thm:main2}
Assume the existence of indistinguishability obfuscation for a class of
puncturable pseudo-random functions.  Then no function family with
super-polynomial pseudo-entropy has an average-case VBB obfuscator with
respect to independent auxiliary input.
\end{theorem}

We describe the specific class for which we need indistinguishability
obfuscation in the proof of the theorem.

Theorems~\ref{thm:main1} and~\ref{thm:main2}, together with
Lemmas~\ref{lemma:worst-case} and~\ref{lemma:avg-case}, immediately yield
impossibility results for VBB obfuscation with a universal simulator. In
particular, Theorem~\ref{thm:main1} and Lemma~\ref{lemma:worst-case} imply
the following corollary.

\begin{corollary}
Assume the existence of a witness encryption scheme for an $\NP$-complete
language. Then no function family with super-polynomial pseudo-entropy has a
VBB obfuscator with a universal simulator.
\end{corollary}

Theorem~\ref{thm:main2} and Lemma~\ref{lemma:avg-case} imply the following
corollary.

\begin{corollary}
Assume the existence of indistinguishability obfuscation for a class of
puncturable pseudo-random functions. Then no function family with
super-polynomial pseudo-entropy has an average-case VBB obfuscator with a
universal simulator.
\end{corollary}

All that remains is to prove Theorems~\ref{thm:main1} and~\ref{thm:main2}.
For notation in both proofs, let $\mathcal{C} = \bigcup_{k \in
\N}\mathcal{C}_k$ be a class of circuits with super-polynomial pseudo-entropy
such that each $C\in \mathcal{C}_k$ maps $\{0,1\}^{\ell(k)}$ to
$\{0,1\}^{\ell'(k)}$. Let $\O$ be any candidate obfuscator for $\mathcal{C}$,
and let $m(k)$ be a polynomial such that $|\O(C)| \leq m(k)$ for every $C \in
\mathcal{C}_k$.

\subsection{Proof of Theorem~\ref{thm:main1}}

The fact that ${\cal C}$ has super-polynomial pseudo-entropy implies that it
has pseudo-entropy at least $m(k)+k$. In particular, recalling
Definition~\ref{def:pseudo-entropy}, this implies that there exists a
polynomial $t=t(k)$ and a subset $I_k\subseteq\{0,1\}^k$ of size $t(k)$ such
that for every $C$ there exists a random variable $Y^C=(Y_1,\ldots,Y_t)$ such
that the following conditions hold:
 \begin{enumerate}
 \item The random variable $Y^C$ has statistical min-entropy at least
     $m(k)+k$.
 \item For every (possibly non-uniform) $\PPT$ distinguisher $\D$,
     \[
     \left|\Pr[\D^{C}(1^k)=1]-\Pr[\D^{C \circ Y^C}(1^k)=1]\right|= \negl(k),
     \]
     where $C \circ Y^C$ denotes an oracle that agrees with $C$ except
that $Y^C$ replaces the values of $C$ for inputs in $I_k$.  Here the
probabilities are over $C\leftarrow{\cal C}_k$, the random variable
$Y^C$, and the random coin tosses of $\D$.
 \end{enumerate}

We define an $\NP$ language ${\cal L}$ by
\[
{\cal L}=\set{(x_i)_{i \in I_k} \pST k \in \N
\mbox{ and there exists a circuit }C \mbox{ of size }|C|\leq p(k) \mbox{ such that }C(i)=x_i \mbox{ for all } i\in I_k}.
\]

Set $x=(C(i))_{i\in I_k}$ and let $\aux(C)=\Enc_x(1^k,b)$, where
$b\leftarrow\{0,1\}$ is a random bit and $\Enc$ is a witness encryption for
the language ${\cal L}$.  Note that the fact that there is a witness
encryption for an $\NP$-complete language implies that there is a witness
encryption for every $\NP$ language, and in particular for ${\cal L}$.

Given $\O(C)$ and $\aux(C)=\Enc_x(1^k,b)$, one can efficiently decrypt~$b$
with probability $1-\negl(k)$, since $\O(C)$ is a valid witness of~$x$. It
remains to prove the following claim.

\begin{claim}\label{claim:we}
For any (possibly non-uniform) $\PPT$ adversary $\Sim$ which takes as input
$\aux(s)=\Enc_x(1^k,b)$ and has black-box access to $C$,
\[
\Pr[\Sim^{C}(\Enc_x(1^k,b))=b]\leq \frac{1}{2}+\negl(k).
\]
\end{claim}

\begin{proof}

Suppose for the sake of contradiction that there exists a $\PPT$ adversary
$\Sim$ such that
\[
\Pr[\Sim^{C}(\Enc_x(1^k,b))=b]\geq \frac{1}{2}+\epsilon(k)
\]
for some non-negligible function~$\epsilon$, where the probability is over
random $C\leftarrow {\cal C}_k$, the choice of $b$, and the randomness of
$\Enc$.

Let $\D$ be the distinguisher that, given oracle access to $C$, does the
following.  First, it computes $x=(C(i))_{i\in I_k}$ by querying the oracle
$t(k)$ times. Then it computes $\Enc_x(1^k,b)$ and simulates
$\Sim^{C}(\Enc_x(1^k,b))$ to arrive at its output.

By assumption,
\[
\Pr[\D^{C}(1^k)=b]\geq \frac{1}{2}+\epsilon(k).
\]
Thus, because ${\cal C}$ has super-polynomial pseudo-entropy,
\begin{equation}\label{eqn:S2}
\Pr[\D^{C \circ Y^C}(1^k)=b]\geq \frac{1}{2}+\epsilon(k)+\negl(k).
\end{equation}
When it is given oracle access to $C\circ Y^C$, $\D$ replaces $x$ with $x^* =
Y^C$, and at the end it is trying to recover $b$ from $\Enc_{x^*} (1^k,b)$.

Note however that $x^*$ has min-entropy $m(k)+k$, and so the probability that
it is in ${\cal L}$ is at most $2^{-k}$.  (For each of the at most $2^{m(k)}$
circuits of size $m(k)$ in the definition of ${\cal L}$, the probability of
obtaining $x^*$ is at most $2^{-m(k)-k}$.) Thus, Equation~\eqref{eqn:S2}
contradicts the semantic security of the underlying witness-encryption
scheme.
\end{proof}

\begin{remark}\label{remark:main}
Note that for any secret predicate~$\pi$ that is not learnable from black-box
access to the circuit, we could have taken the auxiliary input to be
$\aux(C)=\Enc_{x}(1^k,b)$ where $b=\pi(C)$ (as opposed to being truly
random). In this case, there exists a $\PPT$ adversary~$\Adv$ that given the
obfuscated circuit~$\O(C)$ and the auxiliary input~$\aux(C)$ outputs $\pi(C)$
with probability~$1$, whereas any $\PPT$ simulator cannot learn $\pi(C)$ from
$\aux(C)$ and black-box access to~$C$.

Using Lemma~\ref{lemma:worst-case}, we conclude that for any secret
predicate~$\pi$ that is not learnable from black-box access to the circuit
and for any circuit~$C$ there exists an adversary $\Adv_{\aux(C)}$ that
outputs $\pi(C)$ with probability~$1$, whereas any universal simulator
$\Sim$, which is given black box access to~$C$ and takes as input the code
of~$\Adv_{\aux(C)}$, cannot learn the predicate~$\pi(C)$.

Thus our negative result is a strong one: VBB obfuscation with a universal
simulator cannot conceal \emph{any} secret predicate that is not learnable
from black-box access to the circuit.
\end{remark}

\subsection{Proof of Theorem~\ref{thm:main2}}

We first describe an auxiliary-input distribution ensemble $\Z$ and a $\PPT$
adversary $\Adv$ such that given $z\gets\Z$ and an obfuscation of $C\gets
\mathcal{C}$, $\Adv$ always learns some predicate $\pi(C,z)$. Then, we show
that any $\PPT$ simulator that is only given oracle access to $C$ fails to
learn the predicate.

\paragraph{The auxiliary input distribution $\Z$.}
By assumption, $\mathcal{C}$ has pseudo-entropy at least $m(k) + k$. Let
$\set{I_k}_{k\in\N}$ be the sets guaranteed by
Definition~\ref{def:pseudo-entropy}, where $I_k$ is of polynomial size
$t(k)$, and let $\mathcal{G}$ be a puncturable one-bit PRF family
\[
\mathcal{G} = \set{\G_s \colon \{0,1\}^{\ell'(k) \cdot t(k)} \rightarrow \{0,1\} \pST s \in \{0,1\}^k, k\in\N}.
\]

We define two circuit families
\begin{align*}
\mathcal{K} &= \set{K_{s}\colon\{0,1\}^{m(k)}\rightarrow\{0,1\}\pST s \in \{0,1\}^{k}, k\in\N},\\
\mathcal{K}^* &= \set{K^*_{s_{x^*}}\colon\{0,1\}^{m(k)}\rightarrow\{0,1\} \pST s \in \{0,1\}^{k}, x^* \in \{0,1\}^{\ell'(k)\cdot t(k)},k\in\N}.
\end{align*}

Given a circuit $\tC\colon\{0,1\}^{\ell}\rightarrow\{0,1\}^{\ell'}$ of size
$m$ as input, the circuit $K_{s}$ computes $x:=\tC(I_k):= (\tC(i))_{i\in
I_k}$ and outputs $\G_s(x)$.  See Figure~\ref{cfam}.

\protocol {} {The circuit $K_{s}$.} {cfam} {
\begin{description}
\item[Hardwired:] a PRF key $s \in \{0,1\}^k$ and the set $I_k$.
\item[Input:] a circuit
    $\tC\colon\{0,1\}^{\ell}\rightarrow\{0,1\}^{\ell'}$, where
    $|\tC|=m(k)$.
\begin{enumerate}
\item Compute $x=\tC(I_k)$.
\item Return $\G_s(x)$.
\end{enumerate}
\end{description}
}

The circuit $K^*_{s_{x^*}}$, has a hardwired PRF key $s_{x^*}$ that was
derived from $s$ by puncturing it at the point $x^*$. It operates the same as
$K_{s}$, except that when $x=x^*$, it outputs an arbitrary bit, say, $0$. See
Figure~\ref{cfamstar}.  In particular, if $x^* \neq \tC(I_k)$ for all
circuits $\tC\in\{0,1\}^{m(k)}$, then $K^*_{s_{x^*}}$ and $K_{s}$ compute the
exact same function.

\protocol {} {The circuit $K^*_{s_{x^*}}$.} {cfamstar} {
\begin{description}
\item[Hardwired:] a punctured PRF key $s_{x^*}= \punc(s,x^*)$ and the set
    $I_k$.
\item[Input:] a circuit
    $\tC\colon\{0,1\}^{\ell}\rightarrow\{0,1\}^{\ell'}$, where
    $|\tC|=m(k)$.
\begin{enumerate}
\item Compute $x=\tC(I_k)$.
\item If $x\neq x^*$, return $\G_{s_{x^*}}(x)$.
\item If $x= x^*$, return $0$.
\end{enumerate}
\end{description}
}

We are now ready to define our auxiliary-input distribution
$\Z=\set{Z_k}_{k\in\N}$. Let $d=d(k)$ be the maximal size of circuits in
either $\mathcal{K}$ or $\mathcal{K}^*$, corresponding  to security parameter
$k$. Denote by $[K]_d$ a circuit $K$ padded with zeros to size $d$, and by
$[\mathcal{K}]_d$ the class of circuits where every circuit $K \in
\mathcal{K}$ is replaced with $[K]_d$. Let $\iO$ be an indistinguishability
obfuscator for the class $[\mathcal{K}\cup\mathcal{K}^*]_d$.

The distribution $Z_k$ simply consists of an obfuscated (padded) circuit
$K_s$ for a randomly generated $s$.  See Figure~\ref{auxdist}.

\protocol {} {The auxiliary input distribution $Z_k$.} {auxdist} {
\begin{enumerate}
\item Sample $s \gets \mathsf{Gen}_\mathcal{G}(1^k)$.
\item Sample an obfuscation $z \gets \iO([K_s]_{d(k)})$.
\item Output $z$.
\end{enumerate}
}

\paragraph{The adversary $\Adv$ and predicate $\pi$.}
The adversary $\Adv$, given auxiliary input $z = [\iO(K_s)]_{d(k)}$ and an
obfuscation $\O(C)$ with $C\in \mathcal{C}_k$, outputs
\[
z(\O(C)) = K_s(\O(C)) = \G_s(\O(C)(I_k))=\G_s(C(I_k)),
\]
where the above follows by the definition of $K_s$ and the functionality of
$\iO$ and $\O$.

Thus, $\Adv$ always successfully outputs the predicate
\[
\pi(C,K_s) = K_s(C) = \G_s(C(I_k)).
\]

\paragraph{Adversary $\Adv$ cannot be simulated.}
We prove the following claim implying that the candidate obfuscator $\O$ for
the class $\mathcal{C}$ fails to meet the virtual black box requirement:

\begin{claim}\label{claim:impossibility}
For any $\PPT$ simulator $\Sim$,
\[
\Pr_{
\substack{C\gets \mathcal{C}_k\\
z\gets Z_k}}
\left[\Sim^{C}(z)=\pi(C,z)\right] \leq \frac{1}{2}+\negl(k).
\]
\end{claim}

\begin{proof}
Assume towards contradiction that there exists a $\PPT$ simulator $\Sim$ that
learns $\pi(C,z)$ with probability $\frac{1}{2}+\epsilon(k)$, for some
non-negligible $\epsilon$. We show how to use $\Sim$ to break either the
pseudo-entropy of $\mathcal{C}$ or the pseudo-randomness at punctured points
of $\mathcal{G}$.

According to the definition of $Z_k$,
\[
\Pr
\left[\Sim^{C}(\iO([K_s]_d)=\G_s(C(I_k))\right] \geq \frac{1}{2}+\epsilon(k),
\]
where the probability is over $C\gets \mathcal{C}_k$, $s \gets
\mathsf{Gen}_\mathcal{G}(1^k)$, and the random coin tosses of $\Sim$.

Now, for every $C\in \mathcal{C}_k$, let $Y^C=(Y_1,\dots,Y_t)$ be the random
variable guaranteed by the pseudo-entropy of values in $I_k$
(Definition~\ref{def:pseudo-entropy}). We first consider an alternative
experiment in which the oracle $C$ is replaced with an oracle $C\circ Y^C$
that behaves like $C$ on all points outside $I_k$, and on points in $I_k$
answers according to $Y^C$. We claim that
\[
\Pr
\left[\Sim^{C\circ Y^C}(\iO([K_s]_d)=\G_s(Y^C)\right] \geq \frac{1}{2}+\epsilon(k) -\negl(k),
\]
where the probability is over $C\gets \mathcal{C}_k$, the random variable
$Y^C$, $s \gets \mathsf{Gen}_\mathcal{G}(1^k)$, and the coin tosses of
$\Sim$. Indeed, this follows directly from the pseudo-entropy guarantee
(Definition~\ref{def:pseudo-entropy}), together with the fact that a
distinguisher can sample $s$ and compute $\iO([K_s]_d)$ on its own.

Next, we change the above experiment so that instead of an
indistinguishability obfuscation of $K_s$, the simulator gets an
indistinguishability obfuscation of the circuit $K^*_{s_x^*}$, where $s$ is
punctured at the point $x^* =Y^C$. We claim that
\[
\Pr
\left[\Sim^{C\circ Y^C}(\iO([K^*_{s_{x_*}}]_d)=\G_s(Y^C)\right] \geq \frac{1}{2}+\epsilon(k) -\negl(k),
\]
where the probability is over $C\gets \mathcal{C}_k$, the random variable
$Y^C$, $s \gets \mathsf{Gen}_\mathcal{G}(1^k)$, and the coin tosses of
$\Sim$, $x^*=Y^C$, and $s_{x^*}=\punc(s,x^*)$. Indeed, recalling that $Y^C$
has min-entropy $m(k)+k$ for every $C\in \mathcal{C}_k$, there does not exist
a circuit $\tC$ such that $x^*:= Y^C=\tC(I_k)$, except with negligible
probability $2^{-k}$. However, recall that in this case $K_s$ and
$K^*_{s_{x^*}}$ have the exact same functionality, and thus the above follows
by the indistinguishability obfuscation guarantee.

It is now left to note that $\Sim$ predicts with noticeable advantage the
value of $\G_s$ at the punctured point $x^*$, and thus violates the
pseudo-randomness at punctured points requirement
(Definition~\ref{def:punc_prf}).
\end{proof}

\newcommand{\etalchar}[1]{$^{#1}$}


\begin{thebibliography}{GGSW13}

\bibitem[BBC{\etalchar{+}}14]{BarakBCKPS13}
Boaz Barak, Nir Bitansky, Ran Canetti, Yael~Tauman Kalai, Omer Paneth, and Amit
  Sahai.
\newblock Obfuscation for evasive functions.
\newblock In Yehuda Lindell, editor, {\em Theory of Cryptography (TCC 2014)},
  volume 8349 of {\em Lecture Notes in Computer Science}, pages 26--51.
  Springer, 2014.

\bibitem[BGI{\etalchar{+}}01]{BGIRSVY-conf}
Boaz Barak, Oded Goldreich, Russell Impagliazzo, Steven Rudich, Amit Sahai,
  Salil~P. Vadhan, and Ke~Yang.
\newblock On the (im)possibility of obfuscating programs.
\newblock In Joe Kilian, editor, {\em Advances in Cryptology -- CRYPTO 2001},
  volume 2139 of {\em Lecture Notes in Computer Science}, pages 1--18.
  Springer, 2001.

\bibitem[BGI13]{BoyleGI13}
Elette Boyle, Shafi Goldwasser, and Ioana Ivan.
\newblock Functional signatures and pseudorandom functions.
\newblock Cryptology ePrint Archive, Report 2013/401, 2013.
\newblock \url{http://eprint.iacr.org/}.

\bibitem[BGK{\etalchar{+}}13]{BGKPS13}
Boaz Barak, Sanjam Garg, Yael~Tauman Kalai, Omer Paneth, and Amit Sahai.
\newblock Protecting obfuscation against algebraic attacks.
\newblock Cryptology ePrint Archive, Report 2013/631, 2013.
\newblock \url{http://eprint.iacr.org/}.

\bibitem[BR13a]{BR13b}
Zvika Brakerski and Guy~N. Rothblum.
\newblock Black-box obfuscation for $d$-cnfs.
\newblock Cryptology ePrint Archive, Report 2013/557, 2013.
\newblock \url{http://eprint.iacr.org/}.

\bibitem[BR13b]{BR13}
Zvika Brakerski and Guy~N. Rothblum.
\newblock Virtual black-box obfuscation for all circuits via generic graded
  encoding.
\newblock Cryptology ePrint Archive, Report 2013/563, 2013.
\newblock \url{http://eprint.iacr.org/}.

\bibitem[BW13]{BonehW13}
Dan Boneh and Brent Waters.
\newblock Constrained pseudorandom functions and their applications.
\newblock Cryptology ePrint Archive, Report 2013/352, 2013.
\newblock \url{http://eprint.iacr.org/}.

\bibitem[Can97]{Canetti97}
Ran Canetti.
\newblock Towards realizing random oracles: Hash functions that hide all
  partial information.
\newblock In Burton~S. Kaliski, Jr., editor, {\em Advances in Cryptology --
  CRYPTO '97}, volume 1294 of {\em Lecture Notes in Computer Science}, pages
  455--469. Springer, 1997.

\bibitem[CD08]{CD08}
Ran Canetti and Ronny~Ramzi Dakdouk.
\newblock Extractable perfectly one-way functions.
\newblock In Luca Aceto, Ivan Damg{\aa}rd, Leslie~Ann Goldberg, Magn{\'u}s~M.
  Halld{\'o}rsson, Anna Ing{\'o}lfsd{\'o}ttir, and Igor Walukiewicz, editors,
  {\em Proceedings of the 35th International Colloquium on Automata, Languages
  and Programming, Part II}, volume 5126 of {\em Lecture Notes in Computer
  Science}, pages 449--460. Springer, 2008.

\bibitem[CRV10]{CanettiRV10}
Ran Canetti, Guy~N. Rothblum, and Mayank Varia.
\newblock Obfuscation of hyperplane membership.
\newblock In Daniele Micciancio, editor, {\em Theory of Cryptography (TCC
  2010)}, volume 5978 of {\em Lecture Notes in Computer Science}, pages 72--89.
  Springer, 2010.

\bibitem[CV13]{CanettiV13}
Ran Canetti and Vinod Vaikuntanathan.
\newblock Obfuscating branching programs using black-box pseudo-free groups.
\newblock Cryptology ePrint Archive, Report 2013/500, 2013.
\newblock \url{http://eprint.iacr.org/}.

\bibitem[GGH{\etalchar{+}}13]{GGHRSW13}
Sanjam Garg, Craig Gentry, Shai Halevi, Mariana Raykova, Amit Sahai, and Brent
  Waters.
\newblock Candidate indistinguishability obfuscation and functional encryption
  for all circuits.
\newblock In {\em 54th Annual IEEE Symposium on Foundations of Computer Science
  (FOCS 2013)}, pages 40--49. IEEE Computer Society, 2013.

\bibitem[GGJS13]{GGJS13}
Shafi Goldwasser, Vipul Goyal, Abhishek Jain, and Amit Sahai.
\newblock Multi-input functional encryption.
\newblock Cryptology ePrint Archive, Report 2013/727, 2013.
\newblock \url{http://eprint.iacr.org/}.

\bibitem[GGM86]{GoldreichGM86}
Oded Goldreich, Shafi Goldwasser, and Silvio Micali.
\newblock How to construct random functions.
\newblock {\em J. ACM}, 33(4):792--807, 1986.

\bibitem[GGSW13]{GGSW13}
Sanjam Garg, Craig Gentry, Amit Sahai, and Brent Waters.
\newblock Witness encryption and its applications.
\newblock In Dan Boneh, Tim Roughgarden, and Joan Feigenbaum, editors, {\em
  45th Annual ACM Symposium on Theory of Computing (STOC 2013)}, pages
  467--476. ACM, 2013.

\bibitem[GK05]{GK05}
Shafi Goldwasser and Yael~Tauman Kalai.
\newblock On the impossibility of obfuscation with auxiliary input.
\newblock In {\em 46th Annual IEEE Symposium on Foundations of Computer Science
  (FOCS 2005)}, pages 553--562. IEEE Computer Society, 2005.

\bibitem[GKP{\etalchar{+}}13]{GKPVZ13}
Shafi Goldwasser, Yael~Tauman Kalai, Raluca~A. Popa, Vinod Vaikuntanathan, and
  Nickolai Zeldovich.
\newblock Reusable garbled circuits and succinct functional encryption.
\newblock In Dan Boneh, Tim Roughgarden, and Joan Feigenbaum, editors, {\em
  45th Annual ACM Symposium on Theory of Computing (STOC 2013)}, pages
  555--564. ACM, 2013.

\bibitem[GR07]{GR07}
Shafi Goldwasser and Guy~N. Rothblum.
\newblock On best-possible obfuscation.
\newblock In Salil~P. Vadhan, editor, {\em Theory of Cryptography (TCC 2007)},
  volume 4392 of {\em Lecture Notes in Computer Science}, pages 194--213.
  Springer, 2007.

\bibitem[HSW13]{HSW13}
Susan Hohenberger, Amit Sahai, and Brent Waters.
\newblock Replacing a random oracle: Full domain hash from indistinguishability
  obfuscation.
\newblock Cryptology ePrint Archive, Report 2013/509, 2013.
\newblock \url{http://eprint.iacr.org/}.

\bibitem[KPTZ13]{KiayiasPTZ13}
Aggelos Kiayias, Stavros Papadopoulos, Nikos Triandopoulos, and Thomas
  Zacharias.
\newblock Delegatable pseudorandom functions and applications.
\newblock Cryptology ePrint Archive, Report 2013/379, 2013.
\newblock \url{http://eprint.iacr.org/}.

\bibitem[Ore87]{O87}
Yair Oren.
\newblock On the cunning power of cheating verifiers: Some observations about
  zero knowledge proofs.
\newblock In {\em 28th Annual IEEE Symposium on Foundations of Computer
  Science}, pages 462--471. IEEE Computer Society, 1987.

\bibitem[SW13]{SahaiW13}
Amit Sahai and Brent Waters.
\newblock How to use indistinguishability obfuscation: Deniable encryption, and
  more.
\newblock Cryptology ePrint Archive, Report 2013/454, 2013.
\newblock \url{http://eprint.iacr.org/}.

\end{thebibliography}
\end{document}